\documentclass[article]{elsarticle}
\usepackage{mathpazo}
\usepackage{subfig}

\usepackage[T1]{fontenc}
\usepackage[latin9]{inputenc}
\usepackage[usenames,dvipsnames]{color}
\usepackage{float}
\usepackage{amsmath}
\usepackage{amsthm}
\usepackage{amssymb}
\usepackage{graphicx}
\usepackage{setspace}
\usepackage{wasysym}
\usepackage{epstopdf}
\usepackage{mdframed}
\usepackage{xcolor}

\makeatletter
\DeclareMathOperator{\sign}{sign}

\theoremstyle{plain}
\newtheorem{thm}{Theorem}
\newtheorem{lem}[thm]{Lemma}

\newtheorem{cor}{Corollary}

\theoremstyle{definition}

\theoremstyle{remark}

\begin{document}

\title{Gini estimation under infinite variance}

\author{Andrea Fontanari$^1$  - Delft University of Technology and CWI\\
Nassim Nicholas Taleb - Tandon School of Engineering, NYU\\ 
Pasquale Cirillo\footnote{These authors gladly acknowledge the generous support of the EU H2020 Marie Sklodowska-Curie Grant Agreement No 643045 WakEUpCall. Pasquale Cirillo also acknowledges the support of the EU Marie Sklodowska-Curie Career Integration Grant Multivariate Shocks (PCIG13-GA-2013-618794).}$^{,}$\footnote{\textit{Corresponding Author:} P.Cirillo@tudelft.nl. \textit{Address:} Applied Probability Group, EEMCS Faculty, Delft University of Technology, Mekelweg 4, 2628CD Delft, The Netherlands. \textit{Phone:} +31.15.27.82.589.} - Delft University of Technology }

\begin{abstract} {\small
We study the problems related to the estimation of the Gini index in presence of a fat-tailed data generating process, i.e. one in the stable distribution class with finite mean but infinite variance (i.e. with tail index $\alpha\in(1,2)$). We show that, in such a case, the Gini coefficient cannot be reliably estimated using conventional nonparametric methods, because of a downward bias that emerges under fat tails. This has important implications for the ongoing discussion about economic inequality. 

We start by discussing how the nonparametric estimator of the Gini index undergoes a phase transition in the symmetry structure of its asymptotic distribution, as the data distribution shifts from the domain of attraction of a light-tailed distribution to that of a fat-tailed one, especially in the case of infinite variance. We also show how the nonparametric Gini bias increases with lower values of $\alpha$. We then prove that maximum likelihood estimation outperforms nonparametric methods, requiring a much smaller sample size to reach efficiency. Finally, for fat-tailed data, we provide a simple correction mechanism to the small sample bias of the nonparametric estimator based on the distance between the mode and the mean of its asymptotic distribution. 

}
\end{abstract}

\maketitle

\section{Introduction}

Wealth inequality studies represent a field of economics, statistics and econophysics exposed to fat-tailed data generating processes, often with infinite variance \cite{chakra, kk}. This is not at all surprising if we recall that the prototype of fat-tailed distributions, the Pareto, has been proposed for the first time to model household incomes \cite{pareto}. However, the fat-tailedness of data can be problematic in the context of wealth studies, as the property of efficiency (and, partially, consistency) does not necessarily hold for many estimators of inequality and concentration \cite{embrecht, kk}.

The scope of this work is to show how fat tails affect the estimation of one of the most celebrated measures of economic inequality, the Gini index \cite{iddo3, gini, itzaky}, often used (and abused) in the econophysics and economics literature as the main tool for describing the distribution and the concentration of wealth around the world \cite{chakra, chopa, piketty}. 

The literature concerning the estimation of the Gini index is wide and comprehensive (e.g. \cite{iddo3, itzaky} for a review), however, strangely enough, almost no attention has been paid to its behavior in presence of fat tails, and this is curious if we consider that: 1) fat tails are ubiquitous in the empirical distributions of income and wealth \cite{kk, piketty}, and 2) the Gini index itself can be seen as a measure of variability and fat-tailedness \cite{iddosolo, iddo2, iddo2b, fontanari}. 

The standard method for the estimation of the Gini index is nonparametric: one computes the index from the empirical distribution of the available data using Equation \eqref{firstgini} below. But, as we show in this paper, this estimator suffers from a downward bias when we deal with fat-tailed observations. Therefore our goal is to close this gap by deriving the limiting distribution of the nonparametric Gini estimator in presence of fat tails, and propose possible strategies to reduce the bias. We show how the maximum likelihood approach, despite the risk of model misspecification, needs much fewer observations to reach efficiency when compared to a nonparametric one.\footnote{A similar bias also affects the nonparametric measurement of quantile contributions, i.e. those of the type ``the top 1\% owns x\% of the total wealth" \cite{taleb}. This paper extends the problem to the more widespread Gini coefficient, and goes deeper by making links with the limit theorems.}

Our results are relevant to the discussion about wealth inequality, recently rekindled by Thomas Piketty in \cite{piketty, piketty2}, as the estimation of the Gini index under fat tails and infinite variance may cause several economic analyses to be unreliable, if not markedly wrong. Why should one trust a biased estimator?

By fat-tailed data we indicate those data generated by a positive random variable $X$ with cumulative distribution function (c.d.f.) $F(x)$, which is regularly-varying of order $\alpha$ \cite{mikosch}, that is, for $\bar{F}(x):=1-F(x)$, one has
\begin{equation} 
\label{regularyVaring}
\lim_{x\to\infty}x^{\alpha}\bar{F}(x)=L(x),
\end{equation}
where $L(x)$ is a slowly-varying function such that $\lim_{x \to \infty}\frac{L(cx)}{L(x)}=1$ with $c>0$, and where $\alpha > 0$ is called the tail exponent. 

Regularly-varying distributions define a large class of random variables whose properties have been extensively studied in the context of extreme value theory \cite{dehaan, embrecht}, when dealing with the probabilistic behavior of maxima and minima. As pointed out in \cite{cirillo}, regularly-varying and fat-tailed are indeed synonyms. It is known that, if $X_1,...,X_n$ are i.i.d. observations with a c.d.f. $F(x)$ in the regularly-varying class, as defined in Equation \eqref{regularyVaring}, then their data generating process falls into the maximum domain of attraction of a Fr\'echet distribution with parameter $\rho$, in symbols $X\in MDA(\Phi(\rho))$\cite{dehaan}. This means that, for the partial maximum $M_n=\max(X_1,...,X_n)$, one has
\begin{equation} \label{frefre}
P\left(a_n^{-1}\left(M_n-b_n\right)\leq x\right)\overset{d}{\to} \Phi(\rho)=e^{-x^{-\rho}}, \qquad \rho>0,
\end{equation}
with $a_{n}>0$ and $b_n\in \mathbb{R}$ two normalizing constants. Clearly, the connection between the regularly-varying coefficient $\alpha$ and the Fr\'echet distribution parameter $\rho$ is given by: $\alpha=\frac{1}{\rho}$ \cite{embrecht}.\\
The Fr\'echet distribution is one of the limiting distributions for maxima in extreme value theory, together with the Gumbel and the Weibull; it represents the fat-tailed and unbounded limiting case \cite{dehaan}. The relationship between regularly-varying random variables and the Fr\'echet class thus allows us to deal with a very large family of random variables (and empirical data), and allows us to show how the Gini index is highly influenced by maxima, i.e. extreme wealth, as clearly suggested by intuition \cite{fontanari, kk}, especially under infinite variance. Again, this recommends some caution when discussing economic inequality under fat tails.

It is worth remembering that the existence (finiteness) of the moments for a fat-tailed random variable $X$ depends on the tail exponent $\alpha$, in fact
\begin{align}
\label{momentRegVar}
	E(X^{\delta})<\infty & \text{ if } \delta \leq \alpha, \nonumber \\
	E(X^{\delta})=\infty &  \text{ if } \delta>\alpha.
\end{align}
In this work, we restrict our focus on data generating processes with finite mean and infinite variance, therefore, according to Equation \eqref{momentRegVar}, on the class of regularly-varying distributions with tail index $\alpha\in(1,2)$. 

Table \ref{tableofresults} and Figure \ref{fig:empginis} present numerically and graphically our story, already suggesting its conclusion, on the basis of artificial observations sampled from a Pareto distribution (Equation \eqref{pareto} below) with tail parameter $\alpha$ equal to $1.1$. 

Table \ref{tableofresults} compares the nonparametric Gini index of Equation \eqref{firstgini} with the maximum likelihood (ML) tail-based one of Section \ref{mlesec}. For the different sample sizes in Table \ref{tableofresults}, we have generated $10^8$ samples, averaging the estimators via Monte Carlo. As the first column shows, the convergence of the nonparametric estimator to the true Gini value ($g=0.8333$) is extremely slow and monotonically increasing; this suggests an issue not only in the tail structure of the distribution of the nonparametric estimator but also in its symmetry.

Figure \ref{fig:empginis} provides some numerical evidence that the limiting distribution of the nonparametric Gini index loses its properties of normality and symmetry \cite{feller}, shifting towards a skewed and fatter-tailed limit, when data are characterized by an infinite variance. As we prove in Section \ref{sec:main}, when the data generating process is in the domain of attraction of a fat-tailed distribution, the asymptotic distribution of the Gini index becomes a skewed-to-the-right $\alpha$-stable law. This change of behavior is responsible of the downward bias of the nonparametric Gini under fat tails. However, the knowledge of the new limit allows us to propose a correction for the nonparametric estimator, improving its quality, and thus reducing the risk of badly estimating wealth inequality, with all the possible consequences in terms of economic and social policies \cite{kk, piketty, piketty2}.

\begin{table}[h]
\centering
\caption{Comparison of the Nonparametric (NonPar) and the Maximum Likelihood (ML) Gini estimators, using Paretian data with tail $\alpha = 1.1$ (finite mean, infinite variance) and different sample sizes. Number of Monte Carlo simulations: $10^8$. }
\label{tableofresults}
\begin{tabular}{c|cc|cc|c}
\multicolumn{1}{c|}{\textit{n}}                                                                   & \multicolumn{2}{c|}{Nonpar} & \multicolumn{2}{c|}{ML}        & \begin{tabular}[c]{@{}l@{}}Error Ratio\footnote{Error Ratio: $\frac{\sum_{i=1}^{n}|obs_i-(Nonpar)|}{\sum_{i=1}^{n}|obs_i-(ML)|}.$}\end{tabular}                                                       \\
\multicolumn{1}{c|}{\textit{\begin{tabular}[c]{@{}c@{}}(number of obs.)\end{tabular}}} & \textit{Mean}    & \textit{Bias}    & \textit{Mean}   & \textit{Bias} & \\ \hline
\textit{$10^3$}                                                                     & \textit{0.711}   & \textit{-0.122}  & \textit{0.8333} & \textit{0}    & 1.4                                                   \\
\textit{$10^4$}                                                                     & \textit{0.750}   & \textit{-0.083}  & \textit{0.8333} & \textit{0}    & 3                                                     \\
\textit{$10^5$}                                                                     & \textit{0.775}   & \textit{-0.058}  & \textit{0.8333} & \textit{0}    & 6.6                                                   \\
$10^6$                                                                             & 0.790            & -0.043           & 0.8333          & 0             & 156                                                   \\
$10^7$                                                                              & 0.802            & -0.031           & 0.8333          & 0             &    $10^{5}+$
\end{tabular}
\end{table}

\begin{figure*}
\centering{}
\includegraphics[scale=0.75]{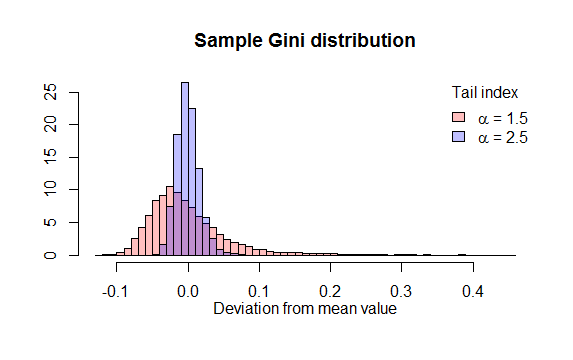}\caption{Histograms for the Gini nonparametric estimators for two Paretian (type I) distributions with different tail indices, with finite and infinite variance (plots have been centered to ease comparison). Sample size: $10^3$. Number of samples: $10^2$ for each distribution. \label{fig:empginis}}
\end{figure*}

The rest of the paper is organized as follows. In Section \ref{sec:main} we derive the asymptotic distribution of the sample Gini index when data possess an infinite variance. In Section \ref{mlesec} we deal with the maximum likelihood estimator; in Section \ref{sec:examples} we provide an illustration with Paretian observations; in Section \ref{sec:correction} we propose a simple correction based on the mode-mean distance of the asymptotic distribution of the nonparametric estimator, to take care of its small-sample bias. Section \ref{conclusions} closes the paper. A technical Appendix contains the longer proofs of the main results in the work.

\section{Asymptotics of the nonparametric estimator under infinite variance}
\label{sec:main}

We now derive the asymptotic distribution for the nonparametric estimator of the Gini index when the data generating process is fat-tailed with finite mean but infinite variance. 

The so-called stochastic representation of the Gini $g$ is 
\begin{equation}
g=\frac{1}{2}\frac{\mathbb{E}\left(|X'-X"|\right)}{\mu} \in [0,1], \label{Gini}
\end{equation}
where $X'$ and $X"$ are i.i.d. copies of a random variable $X$ with c.d.f. $F(x)\in [c,\infty)$, $c>0$, and with finite mean $\mathbb{E}(X)=\mu$. The quantity $\mathbb{E}\left(|X'-X"|\right)$ is known as the "Gini Mean Difference" (GMD) \cite{itzaky}. For later convenience we also define $g=\frac{\theta}{\mu}$ with $\theta=\frac{\mathbb{E}(|X'-X"|)}{2}$.

The Gini index of a random variable $X$ is thus the mean expected deviation between any two independent realizations of $X$, scaled by twice the mean \cite{iddo}.

The most common nonparametric estimator of the Gini index for a sample $X_1,...,X_n$ is defined as
\begin{equation}
\label{firstgini}
	G^{NP}(X_{n})=\frac{\sum_{1\leq i < j \leq n}|X_i - X_j|}{(n-1)\sum_{i=1}^{n}X_i},
\end{equation}
which can also be expressed as
\begin{equation}
G^{NP}(X_{n})= \frac{\sum_{i=1}^{n}(2(\frac{i-1}{n-1}-1)X_{(i)}}{\sum_{i=1}^{n}X_{(i)}} = \frac{\frac{1}{n}\sum_{i=1}^{n}Z_{(i)}}{\frac{1}{n}\sum_{i=1}^{n}X_i},
\label{giniordered}
\end{equation}
where $X_{(1)},X_{(2)},..., X_{(n)}$ are the ordered statistics of $X_1,...,X_n$, such that: $X_{(1)} < X_{(2)} < ... < X_{(n)}$ and $Z_{(i)}=2\left(\frac{i-1}{n-1}-1\right)X_{(i)}$. The asymptotic normality of the estimator in Equation \eqref{giniordered} under the hypothesis of finite variance for the data generating process is known \cite{kk, itzaky}. The result directly follows from the properties of the U-statistics and the L-estimators involved in Equation \eqref{giniordered}

A standard methodology to prove the limiting distribution of the estimator in Equation \eqref{giniordered}, and more in general of a linear combination of order statistics, is to show that, in the limit for $n\to\infty$, the sequence of order statistics can be approximated by a sequence of i.i.d random variables \cite{rao, david}. However, this usually requires some sort of $L^2$ integrability of the data generating process, something we are not assuming here. 

Lemma \ref{lemmaiid} (proved in the Appendix) shows how to deal with the case of sequences of order statistics generated by fat-tailed $L^1$-only integrable random variables.  
\begin{lem}
\label{lemmaiid}
Consider the following sequence $R_n=\frac{1}{n}\sum_{i=1}^{n}(\frac{i}{n}-U_{(i)})F^{-1}(U_{(i)})$ where $U_{(i)}$ are the order statistics of a uniformly distributed i.i.d random sample. Assume that $F^{-1}(U)\in L^1$. Then the following results hold:
\begin{equation}
\label{result_1}
R_{n}\xrightarrow{L^{1}}0,
\end{equation}
and
\begin{equation}
\label{result_2}
\frac{n^{\frac{\alpha-1}{\alpha}}}{L_{0}(n)}R_{n}\xrightarrow{L^{1}}0,
\end{equation}
with $\alpha\in(1,2)$ and $L_{0}(n)$ a slowly-varying function.
\end{lem}

\subsection{A quick recap on $\alpha$-stable random variables}

We here introduce some notation for $\alpha$-stable distributions, as we need them to study the asymptotic limit of the Gini index.

A random variable $X$ follows an $\alpha$-stable distribution, in symbols $X\sim S(\alpha,\beta,\gamma,\delta)$, if its characteristic function is
\begin{equation*}
E(e^{itX})=
\begin{cases}
e^{-\gamma^{\alpha}|t|^{\alpha}(1-i\beta\sign(t))\tan(\frac{\pi\alpha}{2})+i\delta t} & \alpha\ne 1 \\
e^{-\gamma|t|(1+i\beta\frac{2}{\pi}\sign(t))\ln|t|+i\delta t} & \alpha= 1
\end{cases},
\end{equation*}
where $\alpha\in(0,2)$ governs the tail, $\beta\in[-1,1]$ is the skewness, $\gamma\in\mathbb{R}^{+}$ is the scale parameter, and $\delta\in\mathbb{R}$ is the location one. This is known as the $S1$ parametrization of $\alpha$-stable distributions \cite{nolan, taqqu}. 

Interestingly, there is a correspondence between the $\alpha$ parameter of an $\alpha$-stable random variable, and the $\alpha$ of a regularly-varying random variable as per Equation \eqref{regularyVaring}: as shown in \cite{feller, nolan}, a regularly-varying random variable of order $\alpha$ is $\alpha$-stable, with the same tail coefficient. This is why we do not make any distinction in the use of the $\alpha$ here. Since we aim at dealing with distributions characterized by finite mean but infinite variance, we restrict our focus to $\alpha\in (1,2)$, as the two $\alpha$'s coincide. 

Recall that, for $\alpha\in(1,2]$, the expected value of an $\alpha$-stable random variable $X$ is equal to the location parameter $\delta$, i.e. $\mathbb{E}(X)=\delta$. For more details, we refer to \cite{nolan, taqqu}.

The standardized $\alpha$-stable random variable is expressed as
\begin{equation}
\label{Z}
S_{\alpha,\beta}\sim S(\alpha,\beta,1,0).
\end{equation}

We note that $\alpha$-stable distributions are a subclass of infinitely divisible distributions. Thanks to their closure under convolution, they can be used to describe the limiting behavior of (rescaled) partials sums, $S_n=\sum_{i=1}^n X_i$, in the General Central Limit Theorem (GCLT) setting \cite{feller}. For $\alpha=2$ we obtain the normal distribution as a special case, which is the limit distribution for the classical CLTs, under the hypothesis of finite variance.

In what follows we indicate that a random variable is in the domain of attraction of an $\alpha$-stable distribution, by writing $X\in DA(S_{\alpha})$. Just observe that this condition for the limit of partial sums is equivalent to the one given in Equation \eqref{frefre} for the limit of partial maxima \cite{embrecht, feller}.

\subsection{The $\alpha$-stable asymptotic limit of the Gini index} \label{alphasec}

Consider a sample $X_1,...,X_n$ of i.i.d. observations with a continuous c.d.f. $F(x)$ in the regularly-varying class, as defined in Equation \eqref{regularyVaring}, with tail index $\alpha\in(1,2)$. The data generating process for the sample is in the domain of attraction of a Fr\'echet distribution with $\rho\in(\frac{1}{2},1)$, given that $\rho =\frac{1}{\alpha}$. 

For the asymptotic distribution of the Gini index estimator, as presented in Equation \eqref{giniordered}, when the data generating process is characterized by an infinite variance, we can make use of the following two theorems: Theorem \ref{limitGMD} deals with the limiting distribution of the Gini Mean Difference (the numerator in Equation \eqref{giniordered}), while Theorem \ref{limitGini} extends the result to the complete Gini index. Proofs for both theorems are in the Appendix.

\begin{thm}
\label{limitGMD}
Consider a sequence $(X_i)_{1\leq i \leq n}$ of i.i.d random variables from a distribution $X$ on $[c,+\infty)$ with $c>0$, such that $X$ is in the domain of attraction of an $\alpha$-stable random variable, $X\in DA(S_\alpha)$, with $\alpha\in(1,2)$. Then the sample Gini mean deviation (GMD)  $\frac{\sum_{i=1}^{n}Z_{(i)}}{n}$ satisfies the following limit in distribution:
\begin{equation}
\label{GMDstable}
\frac{n^{\frac{\alpha-1}{\alpha}}}{L_{0}(n)}\left(\frac{1}{n}\sum_{i=1}^{n}Z_{(i)}-\theta\right)\overset{d}{\to} S_{\alpha,1},
\end{equation}
where $Z_i=(2F(X_i)-1)X_i$, $\mathbb{E}(Z_{i})=\theta$, $L_{0}(n)$ is a slowly-varying function such that Equation \eqref{L0seq} holds (see the Appendix), and $S_{\alpha,1}$ is a right-skewed standardized $\alpha$-stable random variable defined as in Equation \eqref{Z}.

Moreover the statistic $\frac{1}{n}\sum_{i=1}^{n}Z_{(i)}$ is an asymptotically consistent estimator for the GMD, i.e. $\frac{1}{n}\sum_{i=1}^{n}Z_{(i)}\overset{P}\to \theta$.
\end{thm}
Note that Theorem \ref{limitGMD} could be restated in terms of the maximum domain of attraction $MDA(\Phi(\rho))$ as defined in Equation \eqref{frefre}.

\begin{thm}
\label{limitGini}
Given the same assumptions of Theorem \ref{limitGMD}, the estimated Gini index $G^{NP}(X_{n})=\frac{\sum_{i=1}^{n}Z_{(i)}}{\sum_{i=1}^{n}X_i}$ satisfies the following limit in distribution
\begin{equation}
\label{giniSlim}
\frac{n^{\frac{\alpha-1}{\alpha}}}{L_{0}(n)}\left(G^{NP}(X_{n})-\frac{\theta}{\mu}\right)\overset{d}{\to} Q,
\end{equation}
where $\mathbb{E}(Z_{i})=\theta$, $\mathbb{E}(X_i)=\mu$, $L_{0}(n)$ is the same slowly-varying function defined in Theorem \ref{limitGMD} and $Q$ is a right-skewed $\alpha$-stable random variable $S(\alpha,1,\frac{1}{\mu},0)$.\\
Furthermore the statistic $\frac{\sum_{i=1}^{n}Z_{(i)}}{\sum_{i=1}^{n}X_i}$ is an asymptotically consistent estimator for the Gini index, i.e. $\frac{\sum_{i=1}^{n}Z_{(i)}}{\sum_{i=1}^{n}X_i}\overset{P}\to \frac{\theta}{\mu}=g$.
\end{thm}

In the case of fat tails with $\alpha \in (1,2)$, Theorem \ref{limitGini} tells us that the asymptotic distribution of the Gini estimator is always right-skewed notwithstanding the distribution of the underlying data generating process. Therefore heavily fat-tailed data not only induce a fatter-tailed limit for the Gini estimator, but they also change the shape of the limit law, which definitely moves away from the usual symmetric Gaussian. As a consequence, the Gini estimator, whose asymptotic consistency is still guaranteed \cite{rao}, will approach its true value more slowly, and from below. Some evidence of this was already given in Table \ref{tableofresults}. 

\section{The maximum likelihood estimator} \label{mlesec}

Theorem \ref{limitGini} indicates that the usual nonparametric estimator for the Gini index is not the best option when dealing with infinite-variance distributions, due to the skewness and the fatness of its asymptotic limit. The aim is to find estimators that still preserve their asymptotic normality under fat tails, which is not possible with nonparametric methods, as they all fall into the $\alpha$-stable Central Limit Theorem case \cite{embrecht, feller}. Hence the solution is to use parametric techniques. 

Theorem \ref{MLnormallimit} shows how, once a parametric family for the data generating process has been identified, it is possible to estimate the Gini index via MLE. The resulting estimator is not just asymptotically normal, but also asymptotically efficient. 

In Theorem \ref{MLnormallimit} we deal with random variables $X$ whose distribution belongs to the large and flexible exponential family \cite{shao}, i.e. whose density can be represented as 
$$
f_{\theta}(x)=h(x)e^{\left(\eta(\theta) T(x)-A(\theta)\right)},
$$ with $\theta\in\mathbb{R}$, and where $T(x)$, $\eta(\theta)$, $h(x)$, $A(\theta)$ are known functions.

\begin{thm}
\label{MLnormallimit}
Let $X\sim F_{\theta}$ such that $F_{\theta}$ is a distribution belonging to the exponential family. Then the Gini index obtained by plugging-in the maximum likelihood estimator of $\theta$,  $G^{ML}(X_{n})_{\theta}$, is asymptotically normal and efficient. Namely:
\begin{equation}
\label{giniML}
\sqrt{n}(G^{ML}(X_{n})_{\theta}-g_{\theta})\overset{d}{\to}N(0,g'^{2}_{\theta}I^{-1}(\theta)),
\end{equation}
where $g'_{\theta}=\frac{dg_{\theta}}{d\theta}$ and $I(\theta)$ is the Fisher Information.
\end{thm}

\begin{proof}
The result follows easily from the asymptotic efficiency of the maximum likelihood estimators of the exponential family, and the invariance principle of MLE. In particular, the validity of the invariance principle for the Gini index is granted by the continuity and the monotonicity of $g_{\theta}$ with respect to $\theta$. The asymptotic variance is then obtained by application of the delta-method \cite{shao}.
\end{proof}

\section{A Paretian illustration}
\label{sec:examples}

We provide an illustration of the obtained results using some artificial fat-tailed data. We choose a Pareto I \cite{pareto}, with density
\begin{equation}
\label{pareto}
f(x)=\alpha c^{\alpha } x^{-\alpha -1}	\,, x\geq c. 
\end{equation}
It is easy to verify that the corresponding survival function $\bar{F}(x)$ belongs to the regularly-varying class with tail parameter $\alpha$ and slowly-varying function $L(x)=c^{\alpha}$. We can therefore apply the results of Section \ref{sec:main} to obtain the following corollaries.

\begin{cor}
\label{paretoexampleS}
Let $X_{1},...,X_n$ be a sequence of i.i.d. observations with Pareto distribution with tail parameter $\alpha\in (1,2)$. The nonparametric Gini estimator is characterized by the following limit:
\begin{equation}
\label{eq:approxGinilimit}
D_{n}^{NP}=G^{NP}(X_{n})-g\sim S\left(\alpha,1,\frac{C_\alpha^{-\frac{1}{\alpha}}}{n^{\frac{\alpha-1}{\alpha}}}\frac{(\alpha-1)}{\alpha},0\right).
\end{equation}
\end{cor}

\begin{proof}
Without loss of generality we can assume $c=1$ in Equation \eqref{pareto}. The results is a mere application of Theorem \ref{limitGini}, remembering that a Pareto distribution is in the domain of attraction of $\alpha$-stable random variables with slowly-varying function $L(x)=1$. The sequence $c_n$ to satisfy Equation \eqref{L0seq} becomes $c_n=n^{\frac{1}{\alpha}}C_\alpha^{-\frac{1}{\alpha}}$, therefore we have $L_{0}(n)=C_\alpha^{-\frac{1}{\alpha}}$, which is independent of $n$. Additionally the mean of the distribution is also a function of $\alpha$, that is $\mu=\frac{\alpha}{\alpha-1}$.
\end{proof}

\begin{cor}
\label{paretoexampleN}
Let the sample $X_{1},...,X_n$ be distributed as in Corollary \ref{paretoexampleS}, let $G^{ML}_{\theta}$ be the maximum likelihood estimator for the Gini index as defined in Theorem \ref{MLnormallimit}. Then the MLE Gini estimator, rescaled by its true mean $g$, has the following limit:

\begin{equation}
\label{eq:asydistParetoML}
D^{ML}_n=G^{ML}_{\alpha}(X_{n})-g\sim N\left(0,\frac{4\alpha^2}{n(2\alpha-1)^4}\right),
\end{equation}
where $N$ indicates a Gaussian.
\end{cor}

\begin{proof}
The functional form of the maximum likelihood estimator for the Gini index is known to be $G^{ML}_{\theta}=\frac{1}{2 \alpha^{ML}-1}$ \cite{kk}. The result then follows from the fact that the Pareto distribution (with known minimum value $x_m$) belongs to an exponential family and therefore satisfies the regularity conditions necessary for the asymptotic normality and efficiency of the maximum likelihood estimator. Also notice that the Fisher information for a Pareto distribution is $\frac{1}{\alpha^2}$.
\end{proof}

Now that we have worked out both asymptotic distributions, we can compare the quality of the convergence for both the MLE and the nonparametric case when dealing with Paretian data, which we use as the prototype for the more general class of fat-tailed observations.

In particular, we can approximate the distribution of the deviations of the estimator from the true value $g$ of the Gini index for finite sample sizes, by using Equations \eqref{eq:approxGinilimit} and \eqref{eq:asydistParetoML}. 

\begin{figure*} [htb]
\subfloat[$\alpha=1.8$]{\includegraphics[scale=0.20]{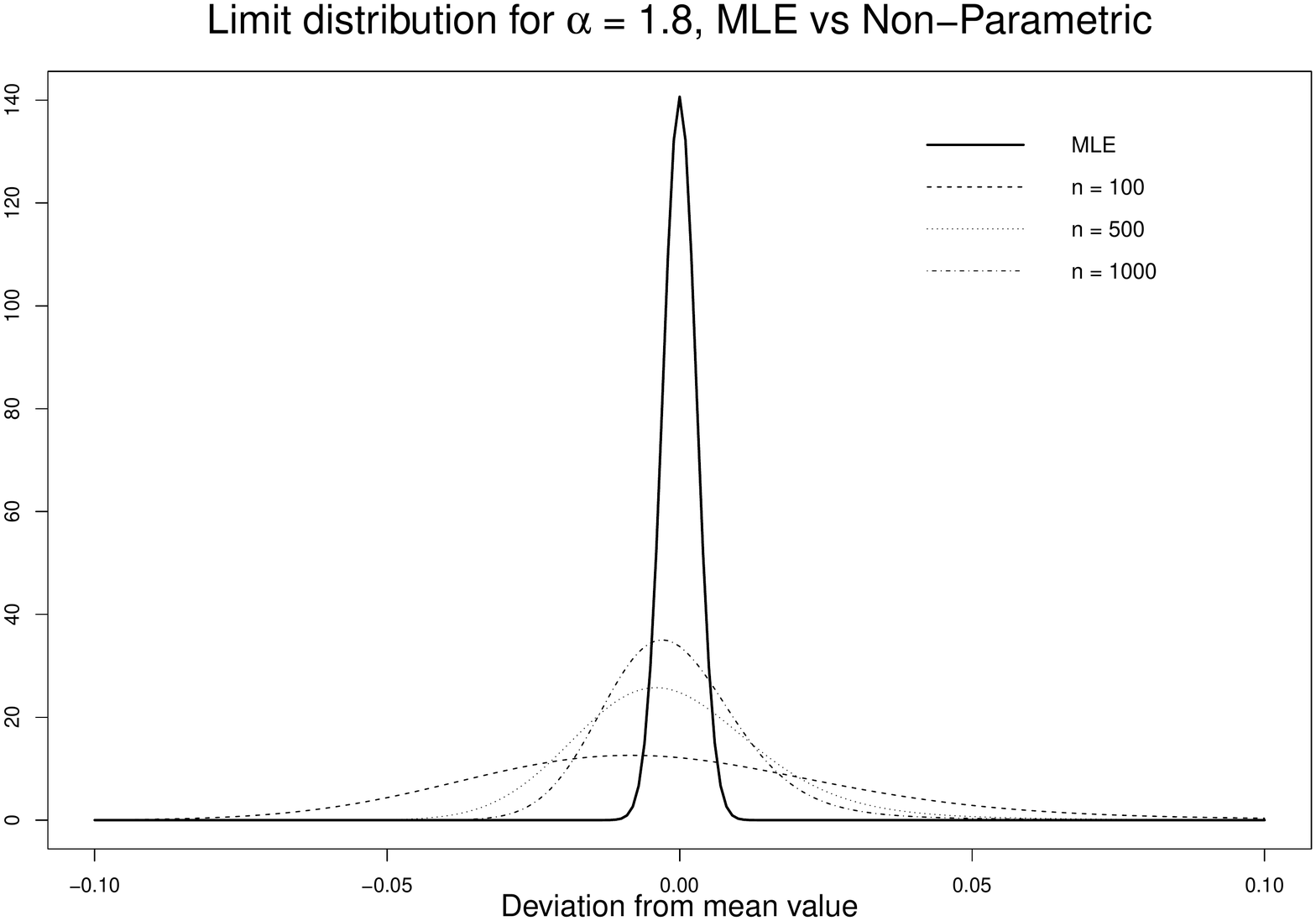}}
\subfloat[$\alpha=1.6$]{\includegraphics[scale=0.20]{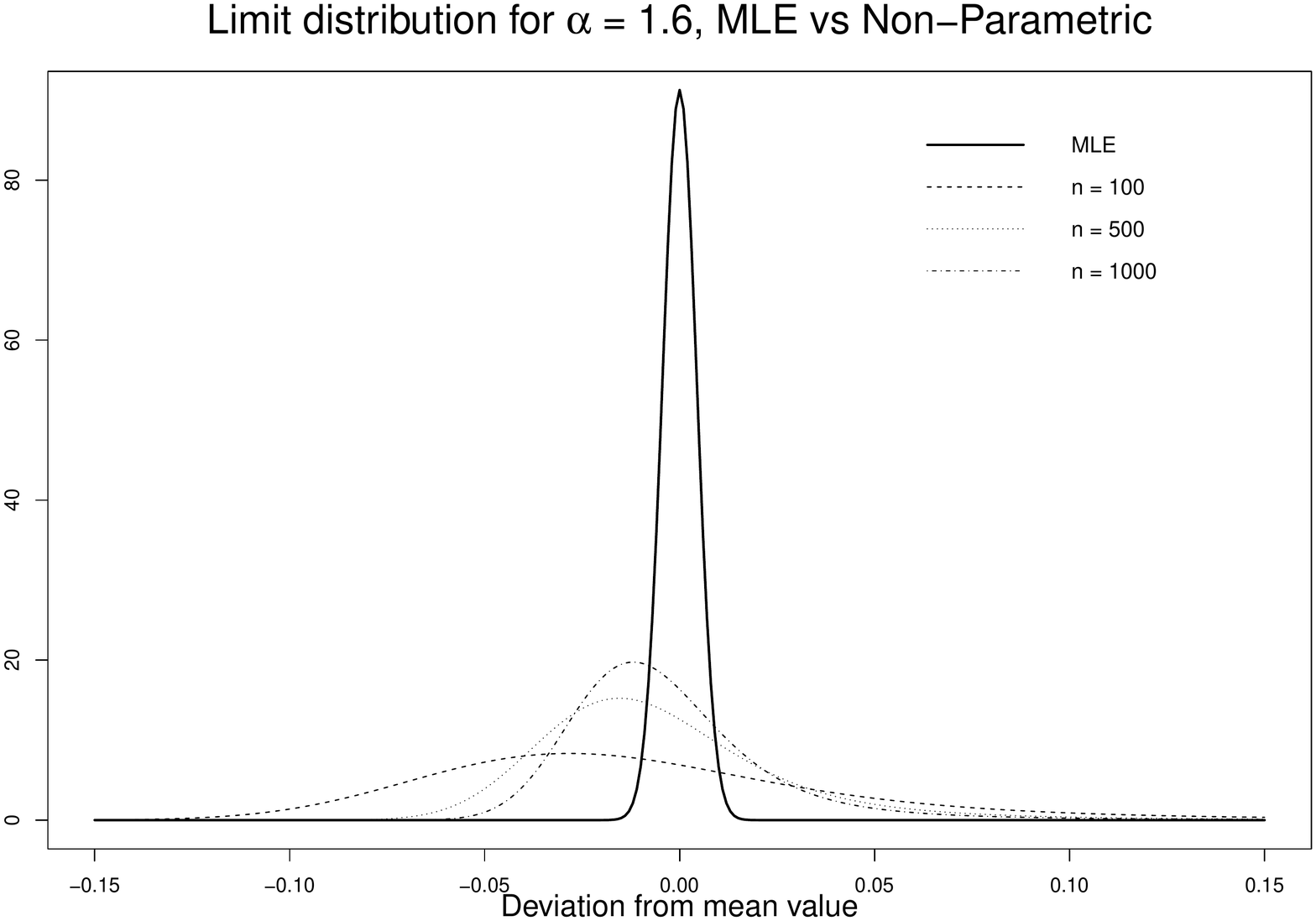}} \\
\subfloat[$\alpha=1.4$]{\includegraphics[scale=0.20]{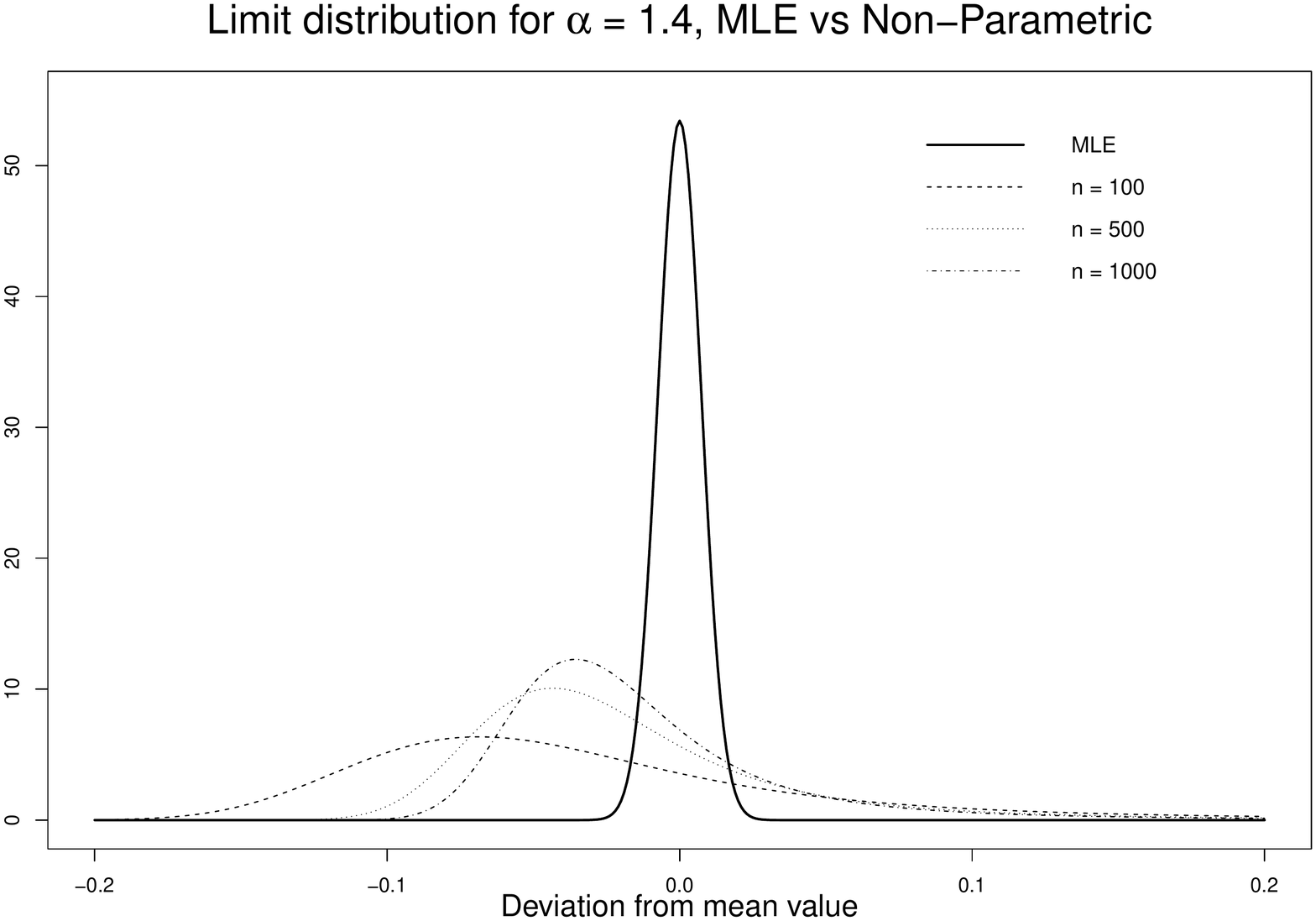}}
\subfloat[$\alpha=1.2$]{\includegraphics[scale=0.20]{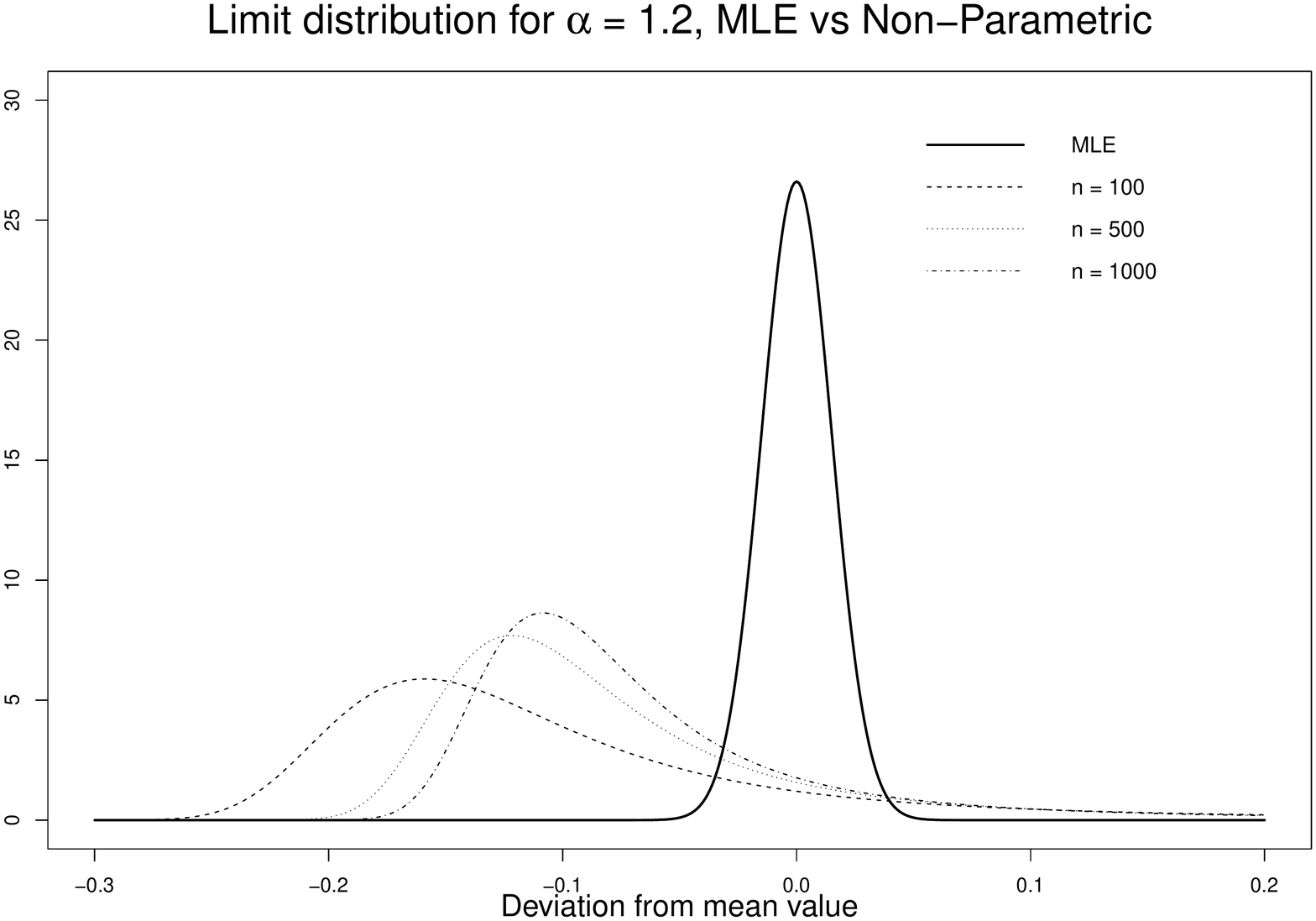}}
\caption{Comparisons between the maximum likelihood and the nonparametric asymptotic distributions for different values of the tail index $\alpha$. The number of observations for MLE is fixed to $n=100$. Note that, even if all distributions have mean zero, the mode of the distributions of the nonparametric estimator is different from zero, because of the skewness.}
\label{normalvs3stable}
\end{figure*}

Figure \ref{normalvs3stable} shows how the deviations around the mean of the two different types of estimators are distributed and how these distributions change as the number of observations increases. In particular, to facilitate the comparison between the maximum likelihood and the nonparametric estimators, we fixed the number of observation in the MLE case, while letting them vary in the nonparametric one. We perform this study for different types of tail indices to show how large the impact is on the consistency of the estimator. It is worth noticing that, as the tail index decreases towards $1$ (the threshold value for a infinite mean), the mode of the distribution of the nonparametric estimator moves farther away from the mean of the distribution (centered on $0$ by definition, given that we are dealing with deviations from the mean). This effect is responsible for the small sample bias observed in applications. Such a phenomenon is not present in the MLE case, thanks to the the normality of the limit for every value of the tail parameter.

We can make our argument more rigorous by assessing the number of observations $\tilde{n}$ needed for the nonparametric estimator to be as good as the MLE one, under different tail scenarios. Let's consider the likelihood-ratio-type function
\begin{equation} \label{likeratio}
r(c,n)=\frac{P_S(|D^{NP}_n|>c)}{P_N(|D^{ML}_{100}|>c)},
\end{equation}
where $P_S(|D^{NP}_n|>c)$ and $P_N(|D^{ML}_{100}|>c)$ are the probabilities ($\alpha$-stable and Gaussian respectively) of the centered estimators in the nonparametric, and in the MLE cases, of exceeding the thresholds $\pm c$, as per Equations \eqref{eq:asydistParetoML} and \eqref{eq:approxGinilimit}. In the nonparametric case the number of observations $n$ is allowed to change, while in the MLE case it is fixed to 100. We then look for the value $\tilde{n}$ such that $r(c,\tilde{n})=1$ for fixed $c$.

Table \ref{table_r(a)} displays the results for different thresholds $c$ and tail parameters $\alpha$. In particular, we can see how the MLE estimator outperforms the nonparametric one, which requires a much larger number of observations to obtain the same tail probability of the MLE with $n$ fixed to 100. For example, we need at least $80\times 10^6$ observations for the nonparametric estimator to obtain the same probability of exceeding the $\pm 0.02$ threshold of the MLE one, when $\alpha=1.2$.

\begin{table}[h]
\centering
\caption{The number of observations $\tilde{n}$ needed for the nonparametric estimator to match the tail probabilities, for different threshold values $c$ and different values of the tail index $\alpha$, of the maximum likelihood estimator with fixed $n=100$.}
\label{table_r(a)}
\begin{tabular}{c|cccc}
	    & \multicolumn{4}{c}{\textit{\begin{tabular}[c]{@{}c@{}}Threshold c as per Equation \eqref{likeratio}: \end{tabular}}} \\
$\alpha$ & 0.005                 & 0.01                  & 0.015                & 0.02                \\ \hline
1.8   & $ 27\times 10^3$                   & $12\times 10^5$                  & $12\times 10^6$                 & $63\times 10^5$               \\
1.5   & $21\times 10^4$                   & $21\times 10^4$                   & $46\times 10^5$                 & $81\times 10^7$                \\
1.2   & $33\times 10^8$                 & $67\times 10^7$                 & $20\times 10^7$               & $80\times 10^6$
\end{tabular}
\end{table}

Interestingly, the number of observations needed to match the tail probabilities in Equation \eqref{likeratio} does not vary uniformly with the threshold. This is expected, since as the threshold goes to infinity or to zero, the tail probabilities remain the same for every value of $n$. Therefore, given the unimodality of the limit distributions, we expect that there will be a threshold maximizing the number of observations needed to match the tail probabilities, while for all the other levels the number of observations will be smaller.

We conclude that, when in presence of fat-tailed data with infinite variance, a plug-in MLE based estimator should be preferred over the nonparametric one.

\section{Small sample correction}
\label{sec:correction}

Theorem \ref{limitGini} can be also used to provide a correction for the bias of the nonparametric estimator for small sample sizes. The key idea is to recognize that, for unimodal distributions, most observations come from around the mode. In symmetric distributions the mode and the mean coincide, thus most observations will be close to the mean value as well, not so for skewed distributions: for right-skewed continuous unimodal distributions the mode is lower than the mean. Therefore, given that the asymptotic distribution of the nonparametric Gini index is right-skewed, we expect that the observed value of the Gini index will be usually lower than the true one (placed at the mean level). We can quantify this difference (i.e. the bias) by looking at the distance between the mode and the mean, and once this distance is known, we can correct our Gini estimate by adding it back\footnote{Another idea, which we have tested in writing the paper, is to use the distance between the median and the mean; the performances are comparable.}.

Formally, we aim to derive a corrected nonparametric estimator $G^{C}(X_{n})$ such that
\begin{equation}
\label{G_corr}
G^{C}(X_{n})=G^{NP}(X_{n})+||m(G^{NP}(X_{n}))-\mathbb{E}(G^{NP}(X_{n}))||,
\end{equation}
where $||m(G^{NP}(X_{n}))-\mathbb{E}(G^{NP}(X_{n}))||$ is the distance between the mode $m$ and the mean of the distribution of the nonparametric Gini estimator $G^{NP}(X_{n})$.

Performing the type of correction described in Equation \eqref{G_corr} is equivalent to shifting the distribution of $G^{NP}(X_{n})$ in order to place its mode on the true value of the Gini index.

Ideally, we would like to measure this mode-mean distance $||m(G^{NP}(X_{n}))-\mathbb{E}(G^{NP}(X_{n}))||$ on the exact distribution of the Gini index to get the most accurate correction. However, the finite distribution is not always easily derivable as it requires assumptions on the parametric structure of the data generating process (which, in most cases, is unknown for fat-tailed data \cite{kk}). We therefore propose to use the limiting distribution for the nonparametric Gini obtained in Section \ref{sec:main} to approximate the finite sample distribution, and to estimate the mode-mean distance with it. This procedure allows for more freedom in the modeling assumptions and potentially decreases the number of parameters to be estimated, given that the limiting distribution only depends on the tail index and the mean of the data, which can be usually assumed to be a function of the tail index itself, as in the Paretian case where $\mu=\frac{\alpha}{\alpha-1}$.

By exploiting the location-scale property of $\alpha$-stable distributions and Equation \eqref{giniSlim}, we approximate the distribution of $G^{NP}(X_{n})$ for finite samples by
\begin{equation}
\label{G_finite_approx}
G^{NP}(X_{n})\sim S\left(\alpha,1,\gamma(n) ,g\right),
\end{equation}
where $\gamma(n) = \frac{1}{n^{\frac{\alpha-1}{\alpha}}}\frac{L_{0}(n)}{\mu}$ is the scale parameter of the limiting distribution.

As a consequence, thanks to the linearity of the mode for $\alpha$-stable distributions, we have
\begin{multline*}
\label{approx_mode}
||m(G^{NP}(X_{n}))-\mathbb{E}(G^{NP}(X_{n}))||\approx||m(\alpha,\gamma(n))+g-g||=||m(\alpha,\gamma(n))||,
\end{multline*}
where $m(\alpha,\gamma(n))$ is the mode function of an $\alpha$-stable distribution with zero mean.

The implication is that, in order to obtain the correction term, knowledge of the true Gini index is not necessary, given that $m(\alpha,\gamma(n))$ does not depend on $g$. We then estimate the correction term as
\begin{equation}
\hat{m}(\alpha,\gamma(n)) = \arg\max_{x} s(x), 
\end{equation}
where $s(x)$ is the numerical density of the associated $\alpha$-stable distribution in Equation \eqref{G_finite_approx}, but centered on $0$. This comes from the fact that, for $\alpha$-stable distributions, the mode is not available in closed form, but it can be easily computed numerically \cite{nolan}, using the unimodality of the law. 

The corrected nonparametric estimator is thus
\begin{equation}
G^{C}(X_{n})=G^{NP}(X_{n})+\hat{m}(\alpha,\gamma(n)),
\end{equation}
whose asymptotic distribution is
\begin{equation}
\label{g_corr_dist}
G^{C}(X_{n})\sim S\left(\alpha,1,\gamma(n),g+\hat{m}(\alpha,\gamma(n))\right).
\end{equation}

Note that the correction term $\hat{m}\left(\alpha,\gamma(n)\right)$ is a function of the tail index $\alpha$ and is connected to the sample size $n$ by the scale parameter $\gamma(n)$ of the associated limiting distribution. It is important to point out that $\hat{m}(\alpha,\gamma(n))$ is decreasing in $n$, and that $\lim_{n\to\infty} \hat{m}(\alpha, \gamma(n)) \to 0$. This happens because, as $n$ increases, the distribution described in Equation \eqref{G_finite_approx} becomes more and more centered around its mean value, shrinking to zero the distance between the mode and the mean. This ensures the asymptotic equivalence of the corrected estimator and the nonparametric one. Just observe that
\begin{eqnarray*}
\lim_{n\to\infty} |G(X_{n})^{C}-G^{NP}(X_{n})| &=&\lim_{n\to\infty}|G^{NP}(X_{n})+\hat{m}(\alpha,\gamma(n))-G^{NP}(X_{n})|\\
&=&\lim_{n\to\infty}|\hat{m}(\alpha,\gamma(n))|\to 0.	
\end{eqnarray*}

Naturally, thanks to the correction, $G^{C}(X_{n})$ will always behave better in small samples. Consider also that, from Equation \eqref{g_corr_dist}, the distribution of the corrected estimator has now for mean $g+\hat{m}(\alpha,\gamma(n))$, which converges to the true Gini $g$ as $n\to\infty$.

From a theoretical point of view, the quality of this correction depends on the distance between the exact distribution of $G^{NP}(X_{n})$ and its $\alpha$-stable limit; the closer the two are to each other, the better the approximation. However, given that, in most cases, the exact distribution of $G^{NP}(X_{n})$ is unknown, it is not possible to give more details.

From what we have written so far, it is clear that the correction term depends on the tail index of the data, and possibly also on their mean. These parameters, if not assumed to be known a priori, must be estimated. Therefore the additional uncertainty due to the estimation will reflect also on the quality of the correction.

We conclude this Section with the discussion of the effect of the correction procedure with a simple example. In a Monte Carlo experiment, we simulate $1000$ Paretian samples of increasing size, from $n=10$ to $n=2000$, and for each sample size we compute both the original nonparametric estimator $G^{NP}(X_{n})$ and the corrected $G^{C}(X_{n})$. We repeat the experiment for different $\alpha$'s. Figure \ref{correction} presents the results.

It is clear that the corrected estimators always perform better than the uncorrected ones in terms of absolute deviation from the true Gini value. In particular, our numerical experiment shows that for small sample sizes with  $n\leq 1000$ the gain is quite remarkable for all the different values of $\alpha\in(1,2)$. However, as expected, the difference between the estimators decreases with the sample size, as the correction term decreases both in $n$ and in the tail index $\alpha$. Notice that, when the tail index equals $2$, we obtain the symmetric Gaussian distribution and the two estimators coincide, given that, thanks to the finiteness of the variance, the nonparametric estimator is no longer biased.

\begin{figure*} 
\subfloat[$\alpha=1.8$]{\includegraphics[scale=0.20]{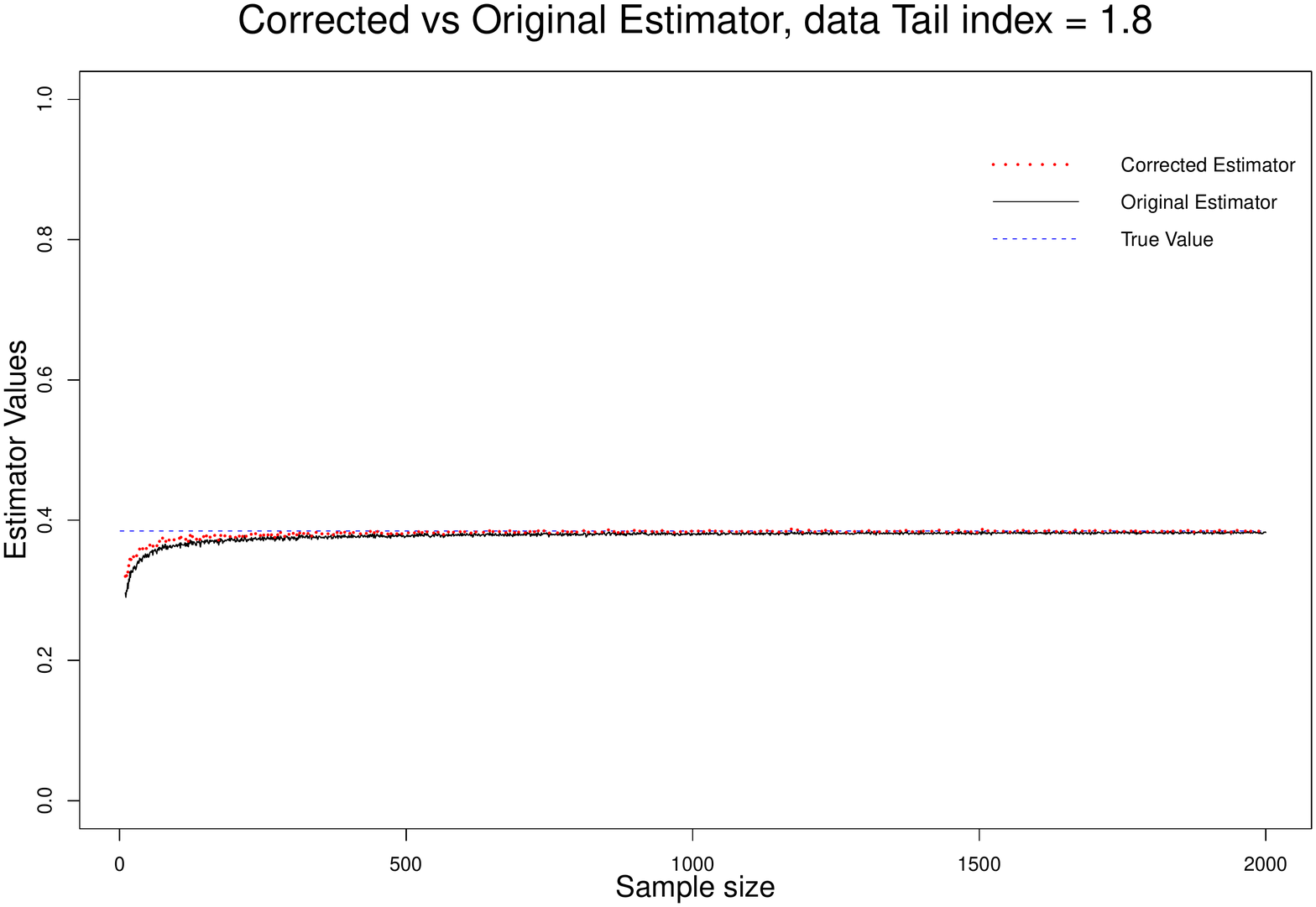}}
\subfloat[$\alpha=1.6$]{\includegraphics[scale=0.20]{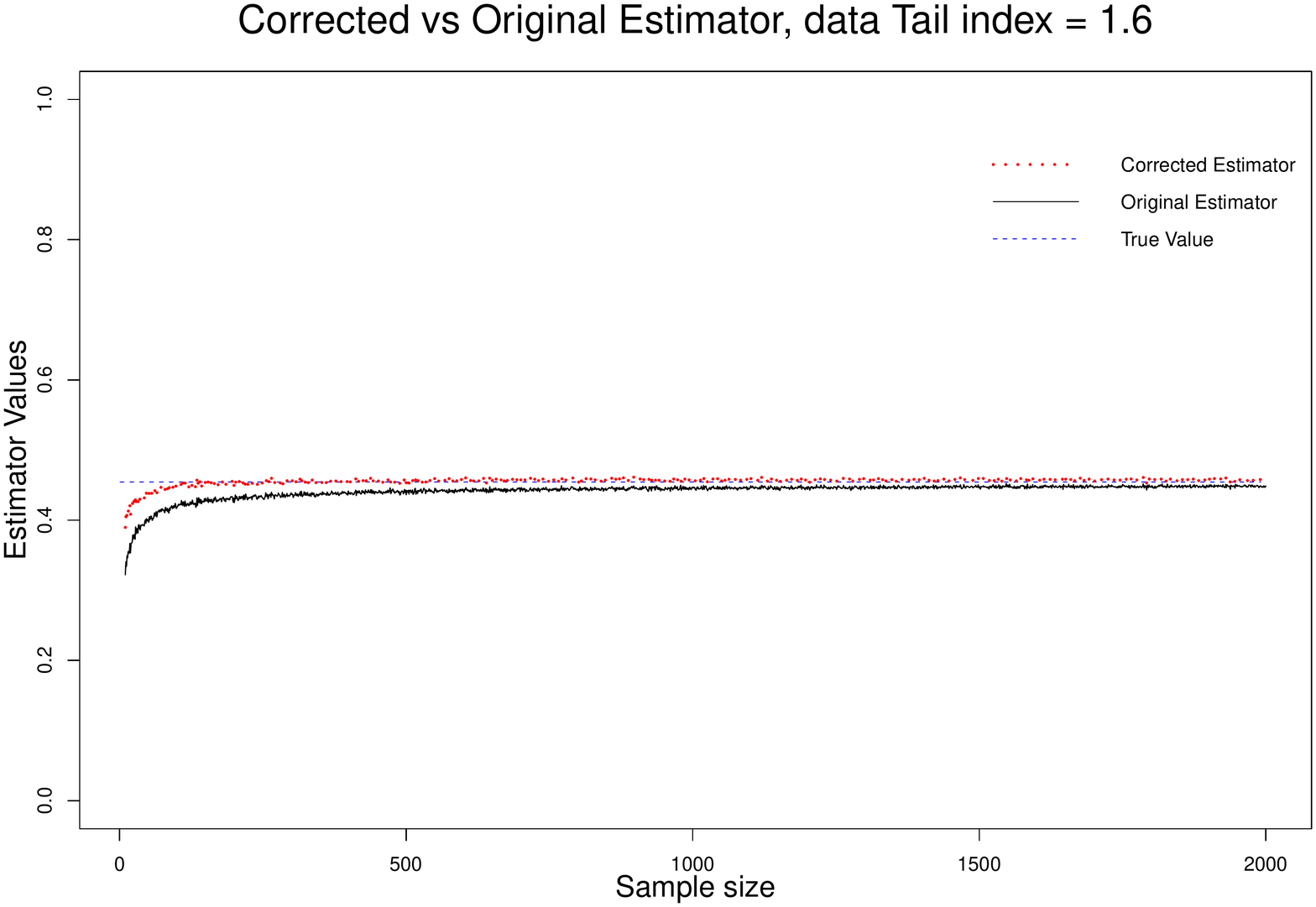}}\\
\subfloat[$\alpha=1.4$]{\includegraphics[scale=0.20]{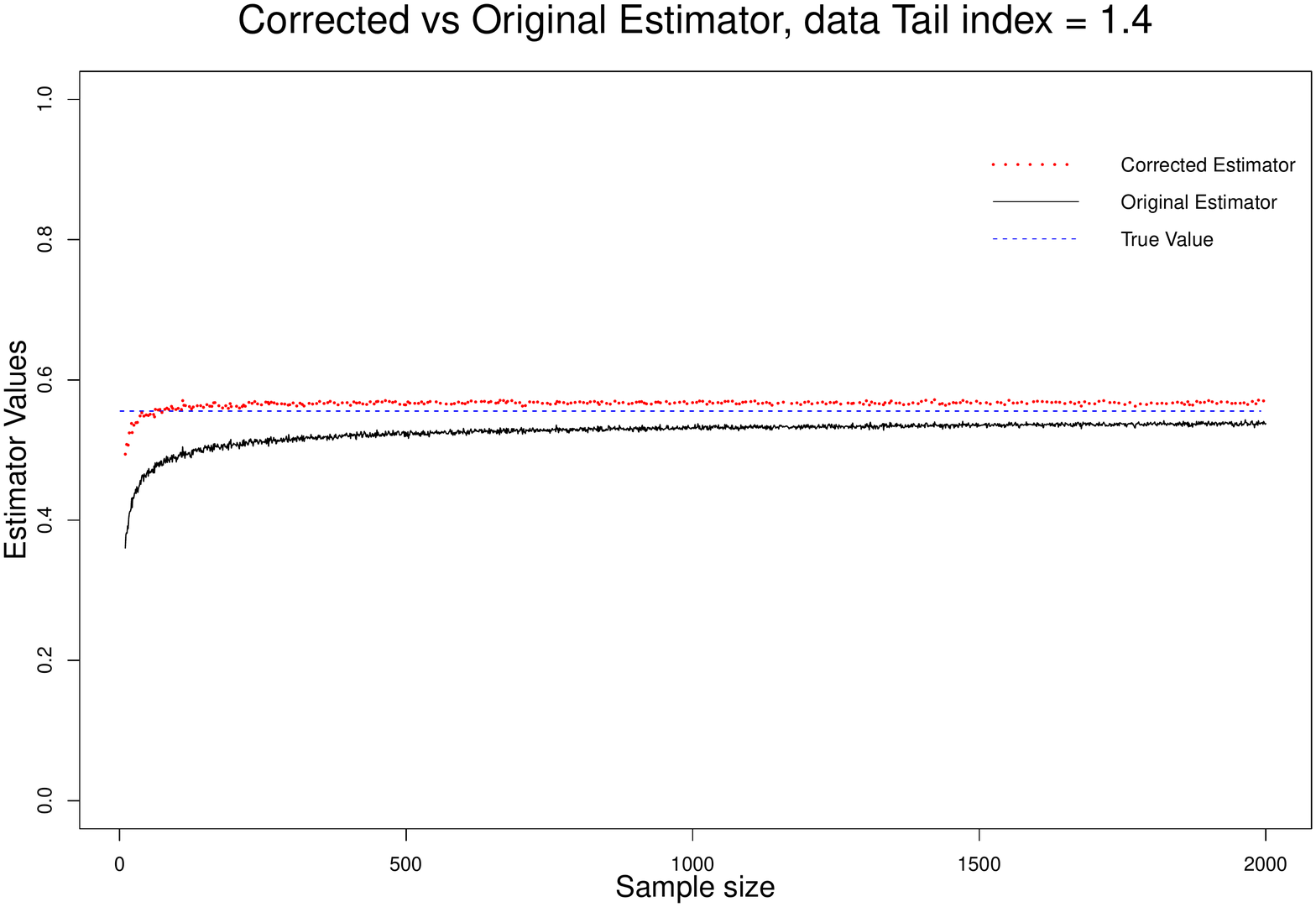}}
\subfloat[$\alpha=1.2$]{\includegraphics[scale=0.20]{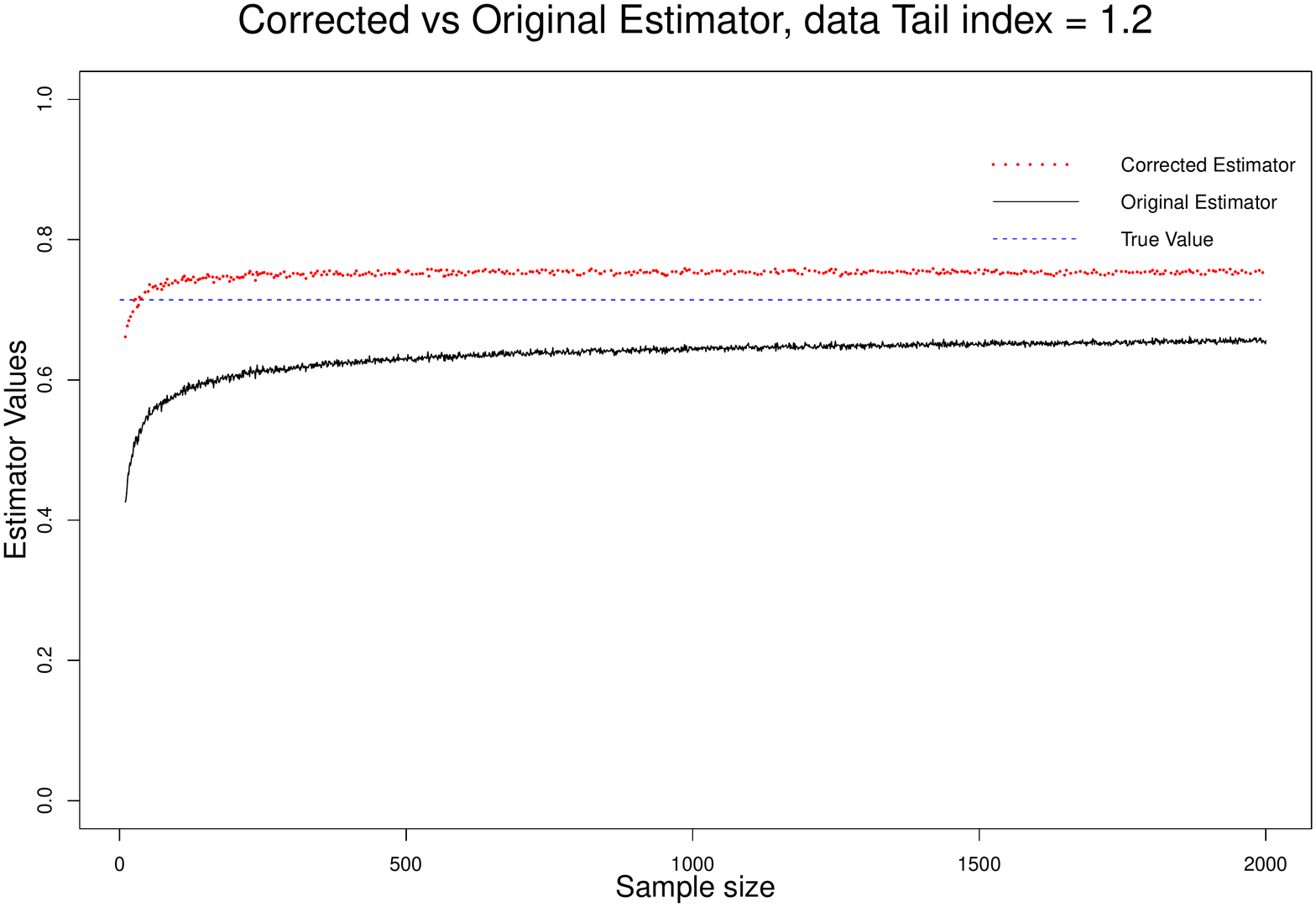}}
\caption{Comparisons between the corrected nonparametric estimator (in red, the one on top) and the usual nonparametric estimator (in black, the one below). For small sample sizes the corrected one clearly improves the quality of the estimation.}
\label{correction}
\end{figure*}

\section{Conclusions} \label{conclusions}

In this paper we address the issue of the asymptotic behavior of the nonparametric estimator of the Gini index in presence of a distribution with infinite variance, an issue that has been curiously ignored by the literature. The central mistake in the nonparametric methods largely used is to believe that asymptotic consistency translates into equivalent pre-asymptotic properties. 

We show that a parametric approach provides better asymptotic results thanks to the properties of maximum likelihood estimation. Hence we strongly suggest that, if the collected data are suspected to be fat-tailed, parametric methods should be preferred.

In situations where a fully parametric approach cannot be used, we propose a simple correction mechanism for the nonparametric estimator based on the distance between the mode and the mean of its asymptotic distribution. Even if the correction works nicely, we suggest caution in its use owing to additional uncertainty from the estimation of the correction term.

\section*{Technical Appendix}

\subsection*{Proof of Lemma \ref{lemmaiid}}

Let $U=F(X)$ be the standard uniformly distributed integral probability transform of the random variable $X$. For the order statistics, we then have \cite{david}: $X_{(i)}\stackrel{a.s.}{=}F^{-1}(U_{(i)})$. Hence
\begin{equation}
\label{Rnexplicit_2}
R_n=\frac{1}{n}\sum_{i=1}^{n}(i/n-U_{(i)})F^{-1}(U_{(i)}).
\end{equation}
Now by definition of empirical c.d.f it follows that
\begin{equation}
\label{finalprep}
R_n=\frac{1}{n}\sum_{i=1}^{n}(F_{n}(U_{(i)})-U_{(i)})F^{-1}(U_{(i)}),
\end{equation}
where $F_{n}(u)=\frac{1}{n}\sum_{i=1}^{n}1_{U_{i}\leq u}$ is the empirical c.d.f of uniformly distributed random variables.

To show that $R_n \xrightarrow{L^{1}}0$, we are going to impose an upper bound that goes to zero. First we notice that
\begin{equation}
\label{firststep}
\mathbb{E}|R_n|\leq \frac{1}{n}\sum_{i=1}^{n}\mathbb{E}|(F_{n}(U_{(i)})-U_{(i)})F^{-1}(U_{(i)})|.
\end{equation}
To build a bound for the right-hand side (r.h.s) of \eqref{firststep}, we can exploit the fact that, while $F^{-1}(U_{(i)})$ might be just $L^{1}$-integrable, $F_{n}(U_{(i)})-U_{(i)}$ is $L^{\infty}$ integrable, therefore we can use H\"older's inequality with $q=\infty$ and $p=1$. It follows that
\begin{equation}
\label{hodler_step}
\frac{1}{n}\sum_{i=1}^{n}\mathbb{E}|(F_{n}(U_{(i)})-U_{(i)})F^{-1}(U_{(i)})| \leq \frac{1}{n}\sum_{i=1}^{n}\mathbb{E}\sup_{U_{(i)}}|(F_{n}(U_{(i)})-U_{(i)})|\mathbb{E}|F^{-1}(U_{(i)})|.
\end{equation}
Then, thanks to the Cauchy-Schwarz inequality, we get
\begin{eqnarray}
\label{cs_step}
&&\frac{1}{n}\sum_{i=1}^{n}\mathbb{E}\sup_{U_{(i)}}|(F_{n}(U_{(i)})-U_{(i)})|\mathbb{E}|F^{-1}(U_{(i)})| \nonumber \\
&&\leq \left( \frac{1}{n}\sum_{i=1}^{n}(\mathbb{E}\sup_{U_{(i)}}|(F_{n}(U_{(i)})-U_{(i)})|)^{2}\frac{1}{n}\sum_{i=1}^{n}(\mathbb{E}(F^{-1}(U_{(i)})))^{2}\right)^{\frac{1}{2}}.
\end{eqnarray}
Now, first recall that $\sum_{i=1}^{n}F^{-1}(U_{(i)})\stackrel{a.s.}{=}\sum_{i=1}^{n}F^{-1}(U_{i})$ with $U_{i}$, $i=1,...,n$, being an i.i.d sequence, then notice that $\mathbb{E}(F^{-1}(U_{i}))=\mu$, so that the second term of Equation \eqref{cs_step} becomes
\begin{equation}
\label{cs_step2}
\mu \left( \frac{1}{n}\sum_{i=1}^{n}(\mathbb{E}\sup_{U_{(i)}}|(F_{n}(U_{(i)})-U_{(i)})|)^{2}\right)^{\frac{1}{2}}.
\end{equation}
The final step is to show that Equation \eqref{cs_step2} goes to zero as $n\to\infty$.

We know that $F_{n}$ is the empirical c.d.f of uniform random variables. Using the triangular inequality the inner term of Equation \eqref{cs_step2} can be bounded as 
\begin{eqnarray}
\label{laststep}
&&\frac{1}{n}\sum_{i=1}^{n}(\mathbb{E}\sup_{U_{(i)}}|(F_{n}(U_{(i)})-U_{(i)})|)^{2} \\
&&\leq \frac{1}{n}\sum_{i=1}^{n}(\mathbb{E}\sup_{U_{(i)}}|(F_{n}(U_{(i)})-F(U_{(i)}))|)^{2} +  \frac{1}{n}\sum_{i=1}^{n}(\mathbb{E}\sup_{U_{(i)}}|(F(U_{(i)})-U_{(i)})|)^{2}. \nonumber
\end{eqnarray}
Since we are dealing with uniforms, we known that $F(U)=u$, and the second term in the r.h.s of \eqref{laststep} vanishes.

We can then bound $\mathbb{E}(\sup_{U_{(i)}}|(F_{n}(U_{(i)})-F(U_{(i)})|)$ using the so called Vapnik-Chervonenkis (VC) inequality, a uniform bound for empirical processes \cite{Dasgupta, empprocess, VC-dim}, getting 
\begin{equation}
\label{VC-ineq}
\mathbb{E}\sup_{U_{(i)}}|(F_{n}(U_{(i)})-F(U_{(i)})| \leq \sqrt[]{\frac{\log(n+1)+\log(2)}{n}}. 
\end{equation}
Combining Equation \eqref{VC-ineq} with Equation \eqref{cs_step2} we obtain
\begin{equation}
\label{grandfinale}
\mu \left( \frac{1}{n}\sum_{i=1}^{n}(\mathbb{E}\sup_{U_{(i)}}|(F_{n}(U_{(i)})-U_{(i)})|)^{2}\right)^{\frac{1}{2}} \leq \mu \, \sqrt{\frac{\log(n+1)+\log(2)}{n}},
\end{equation}
which goes to zero as $n\to\infty$, thus proving the first claim.

For the second claim, it is sufficient to observe that the r.h.s of \eqref{grandfinale} still goes to zero when multiplied by $\frac{n^{\frac{\alpha-1}{\alpha}}}{L_{0}(n)}$ if $\alpha\in(1,2)$.

\subsection*{Proof of Theorem \ref{limitGMD}}

The first part of the proof consists in showing that we can rewrite Equation \eqref{GMDstable} as a function of i.i.d random variables in place of order statistics, to be able to apply a Central Limit Theorem (CLT) argument.

Let's start by considering the sequence 
\begin{equation}
\label{originalseq}
\frac{1}{n}\sum_{i=1}^{n}Z_{(i)}=\frac{1}{n}\sum_{i=1}^{n}\left(2\frac{i-1}{n-1}-1\right)F^{-1}(U_{(i)}).
\end{equation}

Using the integral probability transform $X\stackrel{d}{=}F^{-1}(U)$ with $U$ standard uniform, and adding and removing $\frac{1}{n}\sum_{i=1}^n\left(2U_{(i)}-1 \right)F^{-1}(U_{(i)})$, the r.h.s. in Equation \eqref{originalseq} can be rewritten as
\begin{equation}
\frac{1}{n}\sum_{i=1}^{n}Z_{(i)} = \frac{1}{n}\sum_{i=1}^{n}(2U_{(i)}-1)F^{-1}(U_{(i)}) + \frac{1}{n}\sum_{i=1}^{n} 2\left(\frac{i-1}{n-1}-U_{(i)}\right)F^{-1}(U_{(i)}). 
\end{equation}

Then, by using the properties of order statistics \cite{david} we obtain the following almost sure equivalence
\begin{equation}
\label{decomposition}
\frac{1}{n}\sum_{i=1}^{n}Z_{(i)} \stackrel{a.s.}{=} \frac{1}{n}\sum_{i=1}^{n}(2U_{i}-1)F^{-1}(U_{i}) + \frac{1}{n}\sum_{i=1}^{n} 2\left(\frac{i-1}{n-1}-U_{(i)}\right)F^{-1}(U_{(i)}).
\end{equation}

Note that the first term in the r.h.s of \eqref{decomposition} is a function of i.i.d random variables as desired, while the second term is just a reminder, therefore
\[\frac{1}{n}\sum_{i=1}^{n}Z_{(i)} \stackrel{a.s.}{=}\frac{1}{n}\sum_{i=1}^{n}Z_{i}+R_{n},\]
with $Z_i=(2U_{i}-1)F^{-1}(U_{i})$ and $R_n=\frac{1}{n}\sum_{i=1}^{n}(2(\frac{i-1}{n-1}-U_{(i)}))F^{-1}(U_{(i)})$.

Given Equation \eqref{GMDstable} and exploiting the decomposition given in \eqref{decomposition} we can rewrite our claim as
\begin{equation}
\label{claim_intermediate_step}
\frac{n^{\frac{\alpha-1}{\alpha}}}{L_{0}(n)}\left(\frac{1}{n}\sum_{i=1}^{n}Z_{(i)}-\theta\right)=\frac{n^{\frac{\alpha-1}{\alpha}}}{L_{0}(n)}\left(\frac{1}{n}\sum_{i=1}^{n}Z_{i}-\theta\right) + \frac{n^{\frac{\alpha-1}{\alpha}}}{L_{0}(n)}R_{n}.
\end{equation}

From the second claim of the Lemma \ref{lemmaiid} and Slutsky Theorem, the convergence in Equation \eqref{GMDstable} can be proven by looking at the behavior of the sequence 
\begin{equation}
\frac{n^{\frac{\alpha-1}{\alpha}}}{L_{0}(n)}\left(\frac{1}{n}\sum_{i=1}^{n}Z_{i}-\theta\right),
\end{equation}
where $Z_{i}=(2U_{i}-1)F^{-1}(U_{i})=(2F(X_{i})-1)X_{i}$. This reduces to proving that $Z_{i}$ is in the fat tails domain of attraction.

Recall that by assumption $X\in DA(S_\alpha)$ with $\alpha\in(1,2)$. This assumption enables us to use a particular type of CLT argument for the convergence of the sum of fat-tailed random variables. However, we first need to prove that $Z\in DA(S_{\alpha})$ as well, that is $P(|Z|>z)\sim L(z)z^{-\alpha}$, with $\alpha\in (1,2)$ and $L(z)$ slowly-varying.

Notice that
$$P(|\tilde{Z}|>z)\leq P(|Z|>z)\leq P(2X>z),$$
where $\tilde{Z}=(2U-1)X$ and $U\perp X$. The first bound holds because of the positive dependence between $X$ and $F(X)$ and it can be proven rigorously by noting that $2UX\leq 2F(X)X$ by the so-called re-arrangement inequality \cite{polya}. The upper bound conversely is trivial.

Using the properties of slowly-varying functions, we have  $P(2X>z)\sim 2^{\alpha}L(z)z^{-\alpha}$. To show that $\tilde{Z}\in DA(S_\alpha)$, we use the Breiman's Theorem, which ensure the stability of the $\alpha$-stable class under product, as long as the second random variable is not too fat-tailed \cite{yangwu}.

To apply the Theorem we re-write $P(|\tilde{Z}|>z)$ as
\begin{multline*}
P(|\tilde{Z}|>z)=P(\tilde{Z}>z)+P(-\tilde{Z}>z)=P(\tilde{U}X>z)+P(-\tilde{U}X>z),	
\end{multline*}
where $\tilde{U}$ is a standard uniform with $\tilde{U} \perp X$.

We focus on $P(\tilde{U}X>z)$ since the procedure is the same for $P(-\tilde{U}X>z)$. We have
\begin{multline*}	P(\tilde{U}X>z)=P(\tilde{U}X>z|\tilde{U}>0)P(\tilde{U}>0)+P(\tilde{U}X>z|\tilde{U}\leq0)P(\tilde{U}\leq 0),
\end{multline*}
for $z\to +\infty$.

Now, we have that $P(\tilde{U}X>z|\tilde{U}\leq0)\to 0$, while, by applying Breiman's Theorem, $P(\tilde{U}X>z|\tilde{U}>0)$ becomes
\[P(\tilde{U}X>z|\tilde{U}>0)\to E(\tilde{U}^{\alpha}|U>0)P(X>z)P(U>0).\]
Therefore
\begin{multline*}
	P(|\tilde{Z}|>z)\to \frac{1}{2}E(\tilde{U}^{\alpha}|U>0)P(X>z)+\frac{1}{2}E((-\tilde{U})^{\alpha}|U\leq 0)P(X>z).
\end{multline*}

From this
\begin{eqnarray*}
	P(|\tilde{Z}|>z)&\to& \frac{1}{2}P(X>z)[E(\tilde{U})^{\alpha}|U>0)+E((-\tilde{U}^{\alpha}|U\leq 0)]\\
	&=&\frac{2^\alpha}{1-\alpha}P(X>z)\sim \frac{2^\alpha}{1-\alpha}L(z)z^{-\alpha}.
\end{eqnarray*}

We can then conclude that, by the squeezing Theorem \cite{feller},
\[P(|Z|>z)\sim L(z)z^{-\alpha},\] 
as $z\to \infty$. Therefore $Z\in DA(S_\alpha)$.

We are now ready to invoke the Generalized Central Limit Theorem (GCLT)\cite{embrecht} for the sequence $Z_i$, i.e. 
\begin{equation}
nc_{n}^{-1}\left(\frac{1}{n}\sum_{i=1}^{n}Z_{i}-\mathbb{E}(Z_i)\right)\overset{d}{\to} S_{\alpha,\beta}.
\end{equation}
with $\mathbb{E}(Z_i)=\theta$, $S_{\alpha,\beta}$ a standardized $\alpha$-stable random variable, and where $c_n$ is a sequence which must satisfy
\begin{equation}
\label{L0seq}
\lim_{n\to\infty}\frac{nL(c_n)}{c_n^{\alpha}}=\frac{\Gamma(2-\alpha)|\cos(\frac{\pi\alpha}{2})|}{\alpha-1}=C_\alpha.
\end{equation}
Notice that $c_n$ can be represented as $c_{n}=n^{\frac{1}{\alpha}}L_{0}(n)$, where $L_{0}(n)$ is another slowly-varying function possibly different from $L(n)$.

The skewness parameter $\beta$ is such that 
\[\frac{P(Z>z)}{P(|Z|>z)}\to\frac{1+\beta}{2}.\] Recalling that, by construction, $Z\in [-c,+\infty)$, the above expression reduces to
\begin{equation}
\label{beta}
\frac{P(Z>z)}{P(Z>z)+P(-Z>z)}\to \frac{P(Z>z)}{P(Z>z)}=1\to\frac{1+\beta}{2},
\end{equation}
therefore $\beta=1$. This, combined with Equation \eqref{claim_intermediate_step}, the result for the reminder $R_n$ of Lemma \ref{lemmaiid} and Slutsky Theorem, allows us to conclude that the same weak limits holds for the ordered sequence of $Z_{(i)}$ in Equation \eqref{GMDstable} as well.

\subsection*{Proof of Theorem \ref{limitGini}}

The first step of the proof is to show that the ordered sequence $\frac{\sum_{i=1}^{n}Z_{(i)}}{\sum_{i=1}^{n}X_i}$, characterizing the Gini index, is equivalent in distribution to the i.i.d sequence $\frac{\sum_{i=1}^{n}Z_{i}}{\sum_{i=1}^{n}X_i}$.
In order to prove this, it is sufficient to apply the factorization in Equation \eqref{decomposition} to Equation \eqref{giniSlim}, getting
\begin{equation}
\label{decomposition_gini}
\frac{n^{\frac{\alpha-1}{\alpha}}}{L_{0}(n)}\left(\frac{\sum_{i=1}^{n}Z_i}{\sum_{i=1}^{n}X_i}-\frac{\theta}{\mu}\right)+\frac{n^{\frac{\alpha-1}{\alpha}}}{L_{0}(n)}R_n\frac{n}{\sum_{i=1}^{n}X_i}.
\end{equation}
By Lemma \ref{lemmaiid} and the application of the continuous mapping and Slutsky Theorems, the second term in Equation \eqref{decomposition_gini} goes to zero at least in probability.
Therefore to prove the claim it is sufficient to derive a weak limit for the following sequence
\begin{equation}
\label{gini_proof}
n^{\frac{\alpha-1}{\alpha}}\frac{1}{L_{0}(n)}\left(\frac{\sum_{i=1}^{n}Z_{i}}{\sum_{i=1}^{n}X_{i}}-\frac{\theta}{\mu}\right).
\end{equation}
Expanding Equation \eqref{gini_proof} and recalling that $Z_{i}=(2F(X_{i})-1)X_{i}$, we get
\begin{equation}
\label{gini_expanded}
\frac{n^{\frac{\alpha-1}{\alpha}}}{L_{0}(n)}\frac{n}{\sum_{i=1}^{n}X_{i}}\left(  \frac{1}{n}\sum_{i=1}^{n}X_{i}\left(2F(X_{i})-1-\frac{\theta}{\mu}\right)\right).
\end{equation}
The term $\frac{n}{\sum_{i=1}^{n}X_{i}}$ in Equation \eqref{gini_expanded} converges in probability to $\frac{1}{\mu}$ by an application of the continuous mapping Theorem, and the fact that we are dealing with positive random variables $X$. Hence it will contribute to the final limit via Slutsky Theorem.

We first start by focusing on the study of the limit law of the term
\begin{equation}
\label{target_sequence}
\frac{n^{\frac{\alpha-1}{\alpha}}}{L_{0}(n)}\frac{1}{n}\sum_{i=1}^{n}X_{i}\left(2F(X_{i})-1-\frac{\theta}{\mu}\right).
\end{equation}
Set $\hat{Z}_{i}=X_{i}(2F(X_{i})-1-\frac{\theta}{\mu})$ and note that $\mathbb{E}(\hat{Z}_{i})=0$, since $\mathbb{E}(Z_{i})=\theta$ and $\mathbb{E}(X_i)=\mu$.

In order to apply a GCLT argument to characterize the limit distribution of the sequence $\frac{n^{\frac{\alpha-1}{\alpha}}}{L_{0}(n)}\frac{1}{n}\sum_{i=1}^{n}\hat{Z}_{i}$ we need to prove that $\hat{Z}\in DA(S_\alpha)$. If so then we can apply GCLT to
\begin{equation}
\label{newGCLT}
\frac{n^{\frac{\alpha-1}{\alpha}}}{L_{0}(n)}\left(\frac{\sum_{i=1}^{n}\hat{Z_{i}}}{n}-\mathbb{E}(\hat{Z}_i)\right).
\end{equation}

Note that, since $\mathbb{E}(\hat{Z}_i)=0$, Equation \eqref{newGCLT} equals Equation \eqref{target_sequence}.

To prove that $\hat{Z}\in DA(S_\alpha)$, remember that $\hat{Z_{i}}=X_i(2F(X_{i})-1-\frac{\theta}{\mu})$ is just $Z_{i}=X_{i}(2F(X_i)-1)$ shifted by $\frac{\theta}{\mu}$. Therefore the same argument used in Theorem \ref{limitGMD} for $Z$ applies here to show that $\hat{Z}\in DA(S_\alpha)$. In particular we can point out that $\hat{Z}$ and $Z$ (therefore also $X$) share the same $\alpha$ and slowly-varying function $L(n)$.

Notice that by assumption $X\in [c,\infty)$ with $c>0$ and we are dealing with continuous distributions, therefore $\hat{Z}\in[-c(1+\frac{\theta}{\mu}),\infty)$. As a consequence the left tail of $\hat{Z}$ does not contribute to changing the limit skewness parameter $\beta$, which remains equal to $1$ (as for $Z$) by an application of Equation \eqref{beta}.

Therefore, by applying the GCLT we finally get
\begin{equation}
\label{gini_final}
n^{\frac{\alpha-1}{\alpha}}\frac{1}{L_{0}(n)}(\frac{\sum_{i=1}^{n}Z_{i}}{\sum_{i=1}^{n}X_{i}}-\frac{\theta}{\mu})\xrightarrow{d} \frac{1}{\mu}S(\alpha,1,1,0).
\end{equation}

We conclude the proof by noting that, as proven in Equation \eqref{decomposition_gini}, the weak limit of the Gini index is characterized by the i.i.d sequence of $\frac{\sum_{i=1}^{n}Z_i}{\sum_{i=1}^{n}X_i}$ rather than the ordered one, and that an $\alpha$-stable random variable is closed under scaling by a constant \cite{taqqu}.

\end{document}